\documentclass[11pt,letterpaper]{article}

\usepackage{times}
\usepackage{amsmath,amssymb,amsthm,latexsym,graphicx,epsf,hyperref,psfrag}

\newtheorem{theorem}{Theorem}
\newtheorem{lemma}{Lemma}
\newtheorem{corollary}{Corollary}

\begin{document}

\title{On the Fixed-Parameter Tractability of Some Matching Problems Under the Color-Spanning Model}

\author{Sergey Bereg\footnote
{Department of Computer Science, University of Texas-Dallas, Richardson, TX 75080, USA. Email: \texttt{besp@utdallas.edu}.}
\and
Feifei Ma\footnote{State Key Laboratory of Computer Science, Institute of Software, Chinese Academy of Sciences, Beijing, 100080, China. Email: \texttt{maff@ios.ac.cn}. } 
\and
Wencheng Wang\footnote{State Key Laboratory of Computer Science, Institute of Software, Chinese Academy of Sciences, Beijing, 100080, China. Email: \texttt{whn@ios.ac.cn}. } 
\and
Jian Zhang\footnote{State Key Laboratory of Computer Science, Institute of Software, Chinese Academy of Sciences, Beijing, 100080, China. Email: \texttt{zj@ios.ac.cn}. } 
\and
Binhai Zhu \footnote{Gianforte School of Computing, Montana State University, Bozeman, MT 59717, USA. Email: \texttt{bhz@montana.edu}.}
}

\date{}
\maketitle

\begin{abstract}
Given a set of $n$ points $P$ in the plane, each colored with one of the $t$
given colors, a color-spanning set $S\subset P$ is a subset of $t$ points with
distinct colors. The minimum diameter color-spanning set (MDCS) is a
color-spanning set whose diameter is minimum (among all color-spanning
sets of $P$). Somehow symmetrically, the largest closest pair color-spanning
set (LCPCS) is a color-spanning set whose closest pair is the largest (among all
color-spanning sets of $P$). Both MDCS and LCPCS have been shown to be
NP-complete, but whether they are fixed-parameter tractable (FPT) when $t$ is
a parameter is still open. Motivated by this question, we consider the FPT
tractability of some matching problems under this color-spanning model, where
$t=2k$ is the parameter. The problems are summarized as follows: (1) MinSum Matching
Color-Spanning Set, namely, computing a matching of $2k$ points with distinct
colors such that their total edge length is minimized;
(2) MaxMin Matching Color-Spanning Set, namely, computing a matching of $2k$
points with distinct
colors such that the minimum edge length is maximized;
(3) MinMax Matching Color-Spanning Set, namely, computing a matching of $2k$
points with distinct
colors such that the maximum edge length is minimized; and
(4) $k$-Multicolored Independent Matching, namely, computing a matching of $2k$
vertices in a graph such that the vertices of the edges in the matching do not
share common edges in the graph.
We show that the first three problems are polynomially solvable (hence in FPT), while problem (4) is W[1]-hard.
\end{abstract}


\section{Introduction}

Given a set of $n$ points $Q$ with all points colored in one of the $t$
given colors, a {\em color-spanning} set (sometimes also called a {\em rainbow}
set) is a subset of $t$ points with distinct colors. (In this paper, as we focus
on matching problems, we set $t=2k$. Of course, in general $t$ does not always
have to be even.) In practice, many
problems require us to find a specific color-spanning set with certain property
due to the large size of the color-spanning sets. For instance, in data mining
a problem arises where one wants to find a color-spanning set whose diameter
is minimized (over all color-spanning sets), which can be solved in $O(n^t)$
time using a brute-force method \cite{Zhang,Chen}. (Unfortunately, this is still the best bound
to this date.)

Since the color-spanning set problems were initiated in 2001 \cite{esa01},
quite some related problems have been investigated. Many of the traditional
problems which are polynomially solvable, like Minimum Spanning Tree, Diameter,
Closest Pair, Convex Hull, etc, become NP-hard under the color-spanning model
\cite{rudolf10,rudolf11,ju13}. 
Note that for the hardness results the objective functions are usually slightly changed.
For instance, in the color-spanning model, we would like to maximize the
closest pair and minimize the diameter (among all color-spanning sets).
On the other hand, some problems, like the Maximum Diameter Color-Spanning Set,
remain to be polynomially solvable \cite{fan14}.

In \cite{rudolf10,rudolf11}, an interesting question was raised. Namely, if $t$ is a parameter, is
the NP-complete Minimum Diameter Color-Spanning Set (MDCS) problem
fixed-parameter tractable? This question is still open. In this paper, we
try to investigate some related questions along this line. The base problem
we target at is the matching problem, both under the geometric model and the
graph model. We show that an important graph version is W[1]-hard while all
other versions in consideration are polynomially solvable, hence are
fixed-parameter tractable (FPT).

This paper is organized as follows. In Section 2, we define the basics
regarding FPT algorithms and the problems we will investigate. In Section 3,
we illustrate the positive FPT results on the geometric version MinSum Matching
(and a related graph version). In Section 4, we show the positive results
on the MaxMin Matching and MinMax Matching under the color-spanning model.
In Section 5, we show that a special graph version is W[1]-hard. In Section 5,
we conclude the paper.

\section{Preliminaries}
We make the following definitions regarding this paper.
An Fixed-Parameter Tractable (FPT) algorithm is an algorithm for
a decision problem with input size $n$ and parameter $k$ whose running time
is $O(f(t)n^c)=O^*(f(t))$, where $f(-)$ is any computable function on $t$
and $c$ is a constant. FPT algorithms are efficient tools for handling
some NP-complete problems as they introduce an extra dimension $t$. If
an NP-complete problem, like Vertex Cover, admits an FPT algorithm, then
it is basically polynomially solvable when the parameter $t$ is a small constant
\cite{DF99,jfmg}.

Of course, it is well conceived that not all NP-hard problems admit FPT
algorithms. It has been established that
$$\mbox{FPT}\subseteq W[1]\subseteq W[2]\subseteq \cdots W[z]\subseteq \mbox{XP},$$
where $\mbox{XP}$ represents the set of problem which must take $O(n^t)$ time
to solve (i.e., not FPT), with $t$ being the parameter. Typical problems in W[1] include Independent Set and Clique, etc.
For the formal definition and foundation, readers are referred to
\cite{DF99,jfmg}.

Given a set $Q$ of $n$ points in the plane with $t$ colors, a {\em color-spanning} set
$S\subset Q$ is a subset of $t$ points with distinct colors. If $S$
satisfies a property $\Pi$ among all color-spanning sets of $Q$, we call
the corresponding problem of computing $S$ the Property-$\Pi$ Color-Spanning
Set. For instance, the Minimum Diameter Color-Spanning Set (MDCS) is one where
the diameter of $S$ is minimized (among all color-spanning sets of $Q$) and
the Largest Close Pair Color-Spanning Set (LCPCS) is one where the closest pair
of $S$ is maximized (among all color-spanning sets of $Q$).
All the distances between two points in the plane are Euclidean (or $L_2$). We
next define the matching problems we will investigate in this paper.

Given a set $P$ of $n$ points in the plane with $2k$ colors, a {\em color-spanning} set
$S\subset P$ is a subset of $2k$ points with distinct colors. The points
in $S$ always form a perfect matching, i.e., a set $M$ of $k$ edges connecting
the $2k$ points in $S$. Among all these matchings (over all color-spanning sets),
if a matching $M$ satisfies a property $\Pi$, we call the problem the {\em Property-$\Pi$
Matching Color-Spanning Set} or {\em Property-$\Pi$ Color-Spanning Matching}.
The three properties we focus on are MinSum, MinMax and MaxMin.

MinSum means that the sum of edge lengths in $M$ is minimized,
MinMax means that the maximum edge length in $M$ is minimized, and
MaxMin means that the minimum edge length in $M$ is maximized. The main purpose
of this paper is to investigate the FPT tractability of the three problems:
MinSum Matching Color-Spanning Set, MinMax Matching Color-Spanning Set, and
MaxMin Matching Color-Spanning Set. We show that all these problems are
in fact polynomially solvable (hence FPT).

We also briefly mention some of the related problems on graphs, where we are
given a general weighted graph $G$ whose vertices are colored with $2k$ colors, the
problem is to determine whether a perfect matching $M$ exists such that $M$
contains exactly $2k$ vertices of distinct colors (and if so, compute such a
matching with the minimum weight). We call this problem
$k$-Multicolored Matching, and we will show that this problem is in P (hence FPT).

Finally, we will study a special version on graphs where the (vertices of the)
edges in $M$ cannot share edges in $G$. We call the problem $k$-Multicolored
Independent Matching, and we will show that this problem is W[1]-hard.

\section{MinSum Matching Color-Spanning Set is in P}

In this section, we consider the {\em MinSum Matching Color-Spanning Set} (MSMCS)
problem, namely, given a set $P$ of $n$ points in the plane, each colored with
one of the $2k$ colors, identify $2k$ points with distinct colors such that
they induce a matching with the minimum total weight (among all feasible
color-spanning matchings). Recall that the weight of an edge $(p_i,p_j)$ is the
Euclidean distance between $p_i$ and $p_j$. For a point $p_i$, let $color(p_i)$
be the color of $p_i$. For this problem, we have a useful property of
the optimal solution which is stated as follows.

\begin{lemma}
In an optimal solution of MSMCS, let $p_i$ and $p_j$ be a matched edge in the
optimal matching, then $(p_i,p_j)$ must be the closest pair between points
of $color(p_i)$ and $color(p_j)$.
\end{lemma}

Using this property, we show that the MinSum Matching Color-Spanning Set can be solved in polynomial time (hence FPT).
First, for each ${2k\choose 2}$ pairs of colors, 
compute the bichromatic closest pair of points of the selected colors.
This can be done in $O(n\log n)$ time \cite{PS85} for each pair of colors. 
The total time for all pairs of colors is $O(k^2n\log n)$.
It can be reduced to $O(kn\log n)$ as follows. 
Suppose that the colors are $1,2,\dots,2k$. 
For each $i=1,2,\dots,2k-1$, do the following steps.
\begin{enumerate}
\item[(1)] Make a graph $G=(V,E)$ with $V=\{1,2,\dots,2k\}$ and $E=\emptyset$.
\item[(2)] For points of color $i$, construct the Voronoi diagram and a data structure $D_i$ for point location with $O(\log n)$ query time.
\item[(3)] For each color $j\in\{i+1,i+2,\dots,2k\}$ and each point $p$ of color $j$, find its nearest neighbor $q$ in $D_i$.
For each color $j\in\{i+1,i+2,\dots,2k\}$, compute a pair $(p,q)$ with minimum Euclidean distance and add it to $E$.
\end{enumerate}

Finally, we compute a perfect matching in $G$ of minimum weight using a variation of Edmonds algorithm with running time $O(n^3)$ \cite{lawler76,gabow74}\footnote{Notice that the problems of finding the matchings of minimum weight and of maximum weight are equivalent.}.
We hence have

\begin{theorem}
A minsum matching color-spanning set can be computed in $O(k^{3}+kn\log n)$ time.
\end{theorem}

We next consider the graph version of the MSMCS problem, or,
the $k$-Multicolored Matching problem, which is formally defined as follows.

INSTANCE: An undirected weighted graph $G=(V,E)$ with each vertex colored with one of the $2k$ given colors.

QUESTION: Is there a matching $E' \subseteq E$ including all the $2k$ colors? That is, are there $k$ disjoint edges in $E'$, and all the vertices of the edges in $E'$ have different colors? If such a matching exists, find a minimum weight matching among all such matchings. 


\begin{theorem}
A $k$-multicolored matching can be computed in $O(n+m+k^{3})$ time where $n=|V|$ and $m=|E|$.
\end{theorem}

\begin{proof}
We could simulate the method for Theorem 2 as follows. 
First, construct a weighted graph $G_1=(V_1,E_1)$ where $V_1=\{1,2,\dots,2k\}$ and 
\[
E_1=\{(i,j)~|~\exists (p,q)\in E \text{ with } color(p)=i,color(q)=j\}.
\]
Furthermore, we assign a weight to an edge $(i,j)$ in $E_1$ as the smallest weight of an edge 
$(p,q)$ in $E$ with $color(p)=i$ and $color(q)=j$.
Then the existence of a matching $E'\subseteq E$ is equivalent to the existence of a perfect matching in $G_1$.
The minimum weight matching can be computed using Edmonds algorithm with running time $O(k^{3})$ \cite{lawler76,gabow74}.
The theorem follows since $G_1$ can be constructed in linear time.
\end{proof}

In the next section, we investigate the MaxMin and MinMax Matching Color-Spanning
Sets problems.

\section{MaxMin and MinMax Matching Color-Spanning Sets are in P}

We first study the MaxMin Matching Color-Spanning Set problem, i.e., the
minimum edge length is maximized among all feasible color-spanning matchings.
The first attempt is to try to see whether a property similar to Lemma 1 holds
or not. In Figure 1, we show an example where MaxMin Matching Color-Spanning
Set is not necessarily related to the MinSum (or MaxSum) Color-Spanning Matching. 
In Figure 1, the MinSum Color-Spanning Matching is
$\{(a,c),(b,f)\}$, with a total weight of $2-2\epsilon$.
The MaxSum Color-Spanning Matching is $\{(a,b),(d,e)\}$, which has a total weight
of $1+\sqrt{5}$.
The optimal solution for MaxMin Color-Spanning Matching
 is $\{(a,d),(b,e)\}$, with a solution value of
$\sqrt{2}$ (while the total weight is $2\sqrt{2}$).
Note that $(a,d)$ and $(b,e)$ do not form the closest pairs among
the subsets of respective colors.

For the same point set $\{a,b,c,d,e,f\}$, the color-spanning set $\{a,b,d,e\}$
(which happens to correspond to the point set for MaxMin Color-Spanning
Matching), gives the solution for LCPCS (largest closest pair color-spanning
set). The corresponding closest pair in the set has length 1, while the
solution value for MaxMin Color-Spanning Matching is $\sqrt{2}$. Hence, LCPCS
and MaxMin Color-Spanning Matching are not the same and the claim we made
in the conference version, i.e., LCPCS is FPT \cite{bereg17}, is not correct.

We next show that MaxMin Matching Color-Spanning Set has the following
property.

\begin{lemma}
In an optimal solution of MaxMin Color-Spanning Matching, let $p_i$ and $p_j$
be the minimum matched edge, then $(p_i,p_j)$ must be the farthest pair
between points of $color(p_i)$ and $color(p_j)$.
\end{lemma}
\begin{proof}
Let $d_1(p_i,p_j)$ be the length of the minimum matched edge.
Let $d_2(p_i,p_j)$ be the length of the farthest pair between points of
$color(p_i)$ and $color(p_j)$. Then we could replace 
$d_1(p_i,p_j)$ by $d_2(p_i,p_j)$ to have a new matching whose minimum
matched edge length is longer.
\end{proof}

We next show that MaxMin Matching Color-Spanning Set is polynomially solvable
(hence FPT).
With Lemma 2, we construct a complete graph $G_1$ over $k$ vertices each
corresponding to one of the $k$ colors and between two colors $c_i,c_j$ we have
an edge whose length $w(c_i,c_j)$ is the farthest pair between points of color $c_i$ and $c_j$. 
Similar to Theorem 2, the cost for constructing $G_1$ is $O(kn\log n)$ time.

To solve the problem, we sort all edges in $G_1$. Then for any given edge
$e=(c_i,c_j) \in E(G_1)$, we delete all edges of lengths smaller than $w(e)$
and we delete $c_i,c_j$ as well from $G_1$. Let $G'_1$ be the resulting
graph (containing $2k-2$ colors). Then the problem is to test whether $G'_1$ contains a perfect matching
saturating the remaining $2k-2$ colors. The total cost for this decision
problem is $O(k^{3})$ \cite{lawler76,gabow74}.
We then could use binary search to find the best $e^*$ in
$O(k^{3}\log k)$ time.
The total cost of this algorithm is
$O(k^{3}\log k+kn\log n)$ time.

\begin{figure}[htbp]
\centering
\includegraphics[width=0.33\textwidth]{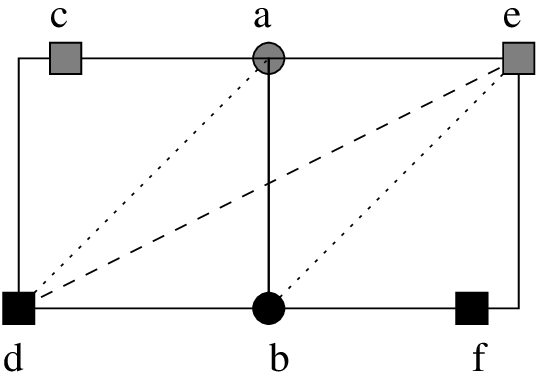}
\label{fig1}
\begin{center}
{\bf Figure 1}. An example of a 4-colored set of 6 points in the plane. 
The edges of both squares have length 1. 
Points $c,f$ are $\epsilon$ distance away from the corresponding closest square corners. 
The MinSum Color-Spanning Matching is $\{(a,c),(b,f)\}$.
The MaxSum Color-Spanning Matching is $\{(a,b),(d,e)\}$.
The MaxMin Color-Spanning Matching is $\{(a,d),(b,e)\}$.
\end{center}
\end{figure}

\begin{theorem}
MaxMin Matching Color-Spanning Set can be solved in 
$O(k^{3}\log k+kn\log n)$ time.
\end{theorem}


Now we consider the MinMax Matching Color-Spanning Set, namely, the maximum 
edge length is minimized among all feasible color-spanning matchings.
Not surprisingly, such a matching might have nothing to do with the 
the MinSum Color-Spanning Matching or the MaxSum Color-Spanning
Matching. In Figure 2, the MinSum Color-Spanning Matching is
$\{(a,b),(c,d)\}$, with a total weight of 3. The
the MaxSum Color-Spanning Matching has a weight at least that of
$\{(a,c),(b,d)\}$ or $\{(c,d),(e,f)\}$, each having a total weight of $4+2\epsilon$.
For the MinMax Color-Spanning Matching problem, all of the above matchings
give a solution value of $2+\epsilon$. 
The optimal solution is $\{(c,e),(d,f)\}$, with a solution value of
$1.5+\epsilon$ (while the total weight is $3+2\epsilon$).
Also, note that $(c,e)$ and $(d,f)$ do not form the farthest pairs among
the subsets of respective colors.

\begin{figure}[htbp]
\centering
\includegraphics[width=0.28\textwidth]{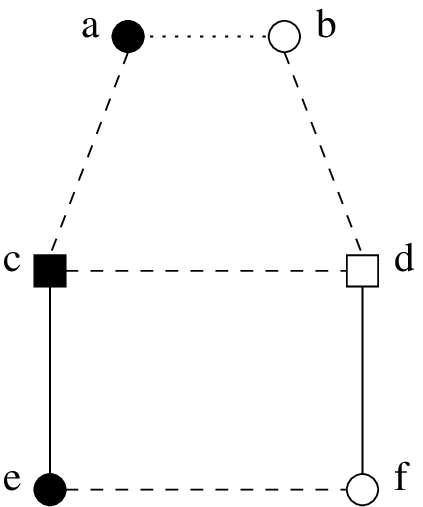}
\label{fig2}
\begin{center}
{\bf Figure 2}. A simple multicolored point set, the dotted, dashed and solid segments have lengths
$1-\epsilon$, $2+\epsilon$ and $1.5+\epsilon$ respectively. The 
MinSum color-spanning matching is $\{(a,b),(c,d)\}$. The
MinMax color-spanning matching is $\{(c,e),(d,f)\}$.
\end{center}
\end{figure}

We next show that MinMax Matching Color-Spanning Set has the following
property.

\begin{lemma}
In an optimal solution of MinMax Color-Spanning Matching, let $p_i$ and $p_j$
be the maximum matched edge, then $(p_i,p_j)$ must be the closest pair
between points of $color(p_i)$ and $color(p_j)$.
\end{lemma}

\begin{proof}
Symmetric to that of Lemma 2, hence omitted.
\end{proof}

We could solve MinMax Color-Spanning Matching in very much the same
way as in Theorem 2, in $O(k^{2.5}\log k+kn\log n)$ time. However, after
a graph $G_2$, over $2k$ colors and the edge weights between two colors
being the closest pair between the corresponding colors, is constructed,
we note that the problem is really the Bottleneck Matching problem on
$G_2$. For a graph with $n_V$ vertices and $n_E$ edges, it is known that such a matching can
be computed in $O(\sqrt{n_V\log n_V}\cdot n_E)$ time \cite{gabow88}. Hence, in our case the MinMax
Color-Spanning Matching can be solved in
$O(k^{2.5}\sqrt{\log k}+kn\log n)$ time.
Therefore we have the following corollary.


\begin{corollary}
MinMax Color-Spanning Matching can be solved in $O(k^{2.5}\sqrt{\log k}+kn\log n)$ time.
\end{corollary}

In the next section, we show that a special matching problem on graphs is in fact
W[1]-hard.

\section{$k$-Multicolored Independent Matching is W[1]-hard}

The $k$-Multicolored Independent Matching problem is defined as follows.

INSTANCE: An undirected graph $G=(V,E)$ with each vertex colored with one of the $2k$ given colors.

QUESTION: Is there an independent matching $E' \subseteq E$ including all the $k$ colors? That is, are there $k$ edges in $E'$ such that
all the vertices of the edges in $E'$ have different colors, and
 for any two edges $(x_1,x_2)$ and $(y_1,y_2)$ in $E'$,
$(x_i,y_j)\not\in E$ (with $i,j=1..2$).

The problem originates from an application in shortwave radio broadcast, where
the matched nodes represent the shortwave channels which should not directly
affect each other \cite{cp16}. (We also comment that this problem seems to be
related to the uncolored version of `Induced Matching' which is known to be
W[1]-hard as well \cite{MS09,MT09}.) We will show that this problem is not
only NP-complete but also W[1]-hard. The problem to reduce from is
the $k$-Multicolored Independent Set, which is defined as follows.

INSTANCE: An undirected graph $G=(V,E)$ with each vertex colored with one of the $k$ given colors.

QUESTION: Is there an independent set $V' \subseteq V$ including all the $k$ colors? That is, are there $k$ vertices in $V'$ incurring no edge in $E$, and all the vertices in $V'$ have different colors.

When $U\subseteq V$ contains exactly $k$ vertices of different colors, we also
say that $U$ is {\em colorful}.

For completeness, we first prove the following lemma, similar to what was done
by Fellows {\em et al.} on $k$-Multicolored Clique problem \cite{Fellows09}.

\begin{lemma}
$k$-Multicolored Independent Set is W[1]-complete.
\end{lemma}

\begin{proof}
The proof can be done through a reduction from $k$-Independent Set.
Given an instance $(G=(V,E),k)$ for $k$-Independent Set, we first make $k$
copies of $G$, $G_i$'s, such that the vertices in each $G_i$ are all colored
with color $i$, for $i=1..k$. For any $u\in V$, let $u_i$ be the corresponding
mirror vertex in $G_i$. Then, for each $(u,v)\in E$ and for each pair of $i,j$,
with $1\leq i\neq j\leq k$, we add four edges $(u_i,u_j)$, $(v_i,v_j)$,
$(u_i,v_j)$ and $(u_j,v_i)$. Let
the resulting graph be $G'$. It is easy to verify that $G$ has a $k$-independent
set if and only if $G'$ has a $k$-multicolored independent set. As $k$-Independent Set is
W[1]-complete \cite{DF99}, the lemma follows.
\end{proof}

The following theorem shows that $k$-Multicolored Independent Matching is not only
NP-complete but also W[1]-hard.

\begin{theorem}
$k$-Multicolored Independent Matching is W[1]-hard, i.e., it does not admit any FPT
algorithm unless FPT=W[1].
\end{theorem}

\begin{proof}
We reduce $k$-Multicolored Independent Set (IS) to the $k$-Multicolored Independent Matching problem.

Given an instance of $k$-Multicolored IS problem, i.e., a graph $G=(V,E)$ with 
each vertex in $V=\{v_1,v_2,...,v_n\}$ colored with one of the $k$ colors
$\{1,2,...,k\}$, the question is whether one could compute an IS of size $k$,
each with a distinct color.

We construct an instance for the $k$-Multicolored Independent Matching as follows.
First, make a copy of $G$ (with the given coloring of $k$ colors). Then,
construct a set $U=\{u_1,u_2,...,u_k\}$ such that $u_i$ has color $k+i$.
Finally, we connect each $u_i\in U$ to each $v_j\in V$ such that
$color(v_j)=i$, i.e., we construct a
set $E'=\{(u_i,v_j)|u_i\in U, v_j\in V, 1\leq i\leq k, 1\leq j\leq n, i=color(v_j)\}$.
(Note that each $u_i\in U$ is connected to nodes in $V$ of exactly one color.)
Let the resulting graph be $G'=(V\cup U, E\cup E')$, with each vertex in $G'$
colored with one of the $2k$ colors.
We claim that $G$ has a colorful independent set of size $k$ if and only if
$G'$ has a colorful independent matching of size $k$. The details are given
as follows.

If $G$ has a colorful independent set $V'\subseteq V$ of size $k$, we select
the $k$ vertices in $V'$ and match them up with the $k$ vertices in $U$ to
obtain $k$ edges (in $E'$). 
The vertices in $V'$ are independent and no two vertices in $U$ share an edge
(i.e., vertices in $U$ are also independent); moreover,
by the definition of $E'$, the vertices of these $k$ edges contain color
pairs $\{(1,k+1),...,(i,k+i),...,(k,2k)\}$. Therefore, among these $k$ edges,
no two edges can have their vertices directly connected (by
edges in $E\cup E'$). Hence, these $k$ edges form a colorful independent
matching for $G'$.

If $G'$ has a colorful independent matching of size $k$, then the $k$ edges
must be obtained by matching exactly $k$ vertices of $V$ with the $k$ vertices
in $U$. (Otherwise, if two vertices $v_i$ and $v_j$ in $V$ form an edge in
the optimal colorful matching then we cannot have $k$ edges in the matching.
This is because at least two vertices in $U$, of colors $color(v_i)+k$ and
$color(v_j)+k$, cannot match up with vertices in $V\cup U$ by the definition
of $E'$. Then the colorful matching contains at most $k-1$ edges, a
contradiction.) By the definition of colorful independent matching, among the
$k$ edges, the $k$ corresponding vertices from $V$ cannot share any edge
hence form an independent set for $G$.

As the reduction takes polynomial time, the theorem is proved.
\end{proof}

We have the following corollary.
\begin{corollary}
The optimization version of $k$-Multicolored Independent Matching (called Multicolored
Maximum Independent Matching) does not admit a factor $n^{1-\epsilon}$ polynomial-time
approximation, for some $\epsilon>0$, unless P=NP.
\end{corollary}

\begin{proof} 
As the reductions in Lemma 4 and Theorem 4 are both L-reductions, the
Multicolored Maximum Independent Matching problem is as hard to approximate as the
Independent Set problem, which does not admit a factor $n^{1-\epsilon}$
polynomial-time approximation, for some $\epsilon>0$, unless P=NP
\cite{Zuc07}.
\end{proof}

\section{Closing Remarks}

Motivated by the open question of Fleischer and Xu, we studied the
FPT tractability of some related matching problems under
the color-spanning model. We show in this paper that most of these
problems are polynomially solvable (hence FPT), except one version on
graphs which can be considered as a generalization of the multicolored
independent set problem. The original question on the FPT tractability of
Minimum Diameter Coloring-Spanning Set (MDCS), is, unfortunately, still open.
Also, contrary to what has been claimed in the conference version
\cite{bereg17}, for the symmetric problem of MDCS, Largest Closest Pair
Color-Spanning Set (LCPCS), its FPT tractability is also open.

\section*{Acknowledgments}

This research is partially supported by NSF of China under project 61628207.
We also thank Ge Cunjing for pointing out some relevant reference, and
an anonymous referee for several useful comments.

\end{document}